 \RequirePackage{fix-cm}
\documentclass[smallextended]{svjour3}       
\smartqed  
\usepackage{graphicx,color,amssymb,latexsym,amsfonts,amsmath}
\newcommand{\be}{\begin{equation}} \newcommand{\ee}{\end{equation}}
\newcommand{\bea}{\begin{eqnarray}} \newcommand{\eea}{\end{eqnarray}}
\newcommand{\beaa}{\begin{eqnarray*}} \newcommand{\eeaa}{\end{eqnarray*}}
\newcommand{\ben}{\begin{enumerate}} \newcommand{\een}{\end{enumerate}}

\begin{document}

\title{On the Establishment, Persistence, and Inevitable Extinction of
  Populations}



\titlerunning{The Life Time of Populations}        

\author{Kais Hamza \and Peter Jagers \and Fima C. Klebaner}


\institute{Kais Hamza \at
              School of Mathematical Sciences, Monash University, Clayton, 3058.Vic. Australia.\\
              \email{kais.hamza@sci.monash.edu.au}                   
           \and Peter Jagers \at Mathematical Sciences, Chalmers
           University of Technology and University of Gothenburg,
           SE-412 96 Gothenburg, Sweden.\\
           \email{jagers@chalmers.se}
           \and
           Fima C. Klebaner \at
             School of Mathematical
Sciences, Monash University, Clayton, 3058 Vic. Australia. \\
              \email{fima.klebaner@sci.monash.edu.au}                   }


\maketitle

\begin{abstract}
Comprehensive models of stochastic, clonally reproducing populations are defined in terms of general branching processes, allowing birth during maternal life, as for higher organisms, or by splitting, as in cell division. The populations are assumed to start small, by mutation or immigration, reproduce supercritically while smaller than the habitat carrying capacity but subcritically above it. Such populations establish themselves with a probability wellknown from branching process theory. Once established, they grow up to a band around the carrying capacity in a time that is logarithmic in the latter, assumed large. There they prevail during a time period whose duration is exponential  in the carrying capacity. Even populations whose life style is sustainble in the sense that the habitat carrying capacity is not eroded but remains the same,  ultimately enter an extinction phase, which again lasts for a time logarithmic in the carrying capacity. However, if the habitat can carry a population which is large, say millions of individuals, and it manages to avoid early extinction, time in generations to extinction will be exorbitantly long, and during it,  population composition over ages, types, lineage etc. will have time to stabilise. This paper  aims at an exhaustive description of the life cycle of such populations, from inception to extinction, extending and overviewing earlier results. We shall also say some words on persistence times of populations with smaller carrying capacities and short life cycles, where the population may indeed be in danger in spite of not eroding its environment.
\end{abstract}

\section{Introduction}
\label{intro}
From a biological viewpoint branching processes are often described as modelling the initial stage of population development, while the population is so small that it is not affected by environmental limitations, but chance events in indivudual life may decide the fate of the whole population. Thus, the classical Galton-Watson extinction problem emerges as the question of early extinction, and one minus the extinction probability as the establishment or invasion probability, cf. \cite{metz96}, \cite{dl}, \cite{waxman}. Classical theory also yields the Malthusian dichotomy, that populations unhampered by environmental limitations, either die out or else grow exponentially \cite{pjbook}. In addition, general branching process theory provides the rate of this growth to infinity and exhibits the stabilisation of composition (\cite{pjbook}, \cite{haccou}) thus underpinning and generalising classical demographic stable population theory, but also showing that relations between individuals, like phylogenies of typical individuals, stabilise \cite{jn}. 

It is common sense that in a finite habitat growth has its limits. These matters were first approached by simplistic macro models, like those of logistic growth. In a more sophisticated, but still deterministic, context they have been analysed  through the individual based approach of structured population dynamics  \cite{cushing}, \cite{odo1}, \cite{odo2}, \cite{webb}, generalised into adaptive dynamics, allowing also mutation and competition, matters we shall disregard here. In such a framework, it was shown that a population not dying out will stabilise at the carrying capacity. From a probabilistic aspect it is however clear that if there is realistic variation between individuals, all  populations  (not growing beyound all bounds) must die out \cite{pjstab}, even those where there are intricate patterns of interaction, dependence loops, or external effects - exempting of course artifacts like immigration from never-ceasing sources. Essentially the sole exception is furnished by ``populations''  where all individuals always beget exactly one child.

The first formulations of probabilistic models with population-size-dependence are due to Klebaner  \cite{JK}, whereas probabilistic formulations of adaptive dynamics have been given in \cite{meleard}, \cite{tran}. 

The notion of  a {\it carrying capacity} plays a great role in biological population dynamics. It is folklore that if a virgin population does not succumb quickly, it should grow  exponentially until its size approaches the capacity of the habitat. Then population size seems to stabilise, at least for the time being. Such  patterns underlie many phenomena in evolutionary biology and ecology,  in particular in adaptive dynamics, \cite{meleard}, \cite{metz}, \cite{meltran}, \cite{tran}, \cite{waxman}, just to mention a few titles in the vast literature.

In an earlier paper we studied this structure in terms of a simple but illustrative Galton-Watson type process with binary splitting, \cite{ksvhj}. Let $K$ denote the carrying capacity and assume that the probability of an individual splitting into two is  $p(z)=K/(K+z)$, if population size is $z$. Otherwise she gets no children. Clearly, the mean reproduction  $2K/(K+z)$ decreases in $z$ and passes 1 precisely at the carrying capacity; the process is supercritical below and subcritical above $K$.

In a sequel \cite{lingering}, this toy model was rendered more realistic,  by replacing binary splitting and deterministic life spans by general, population-size dependent reproduction and general life span distibutions.

In such processes  life-spans may be be influenced by population size, through a  hazard rate  $h_z^K(a)$ of an $a$-aged individual in a population of size $z$. (Not to complicate matters we mostly take life-spans to be continuous.)  Offspring at death (splitting) may be possible, the number of children at splitting depending upon population size $z$ as well as mother's age $a$ at death,   with expectation $m_z^K(a)$, and (female) individuals can give birth according to
age-specific birth rates, $b_z^K(a)$, now to be taken as dependent on the size $z$ of the population at the time when the individual is of age $a$. What renders the thus resulting processes amenable to analysis
is that they are {\em Markovian in the age structure}, \cite{pjbook}, p. 208.

The broadest possible framework would be completely general branching processes, supercritical below and subcritical above the carrying capacity, though some care has to be exercised in such descriptions, since the age-distribution plays a role for the fertility of the population. 



 We give a rigorous formulation of age- and population-size dependent processes which are Markovian in the age structure and have well defined intensities of birth and death. This framework encompasses virtually all classical population dynamics, like age-dependent branching processes as well as birth-and-death processes with age-dependent intensities, and various deterministic approaches, like age structured population dynamics. In the case of Bellman-Harris processes, where reproduction distributions are unaffected by mother's age at death (but not of population size), conditions simplify substantially.

Thus, imagine  a collection of individuals with ages
$(a^1,\ldots,a^{z})=A$, $z=|A|$ being the number of elements in $A$.
It is convenient to regard the collection of ages $A$ as a measure
$$A=\sum_{i=1}^z\delta_{a^i},$$
where $\delta_a$ denotes the point measure at $a$. As usual, the
following notation is  used for a function $f$ on  ${\Bbb R}$ and measure $A$:
$$(f,A)=\int f(a)A(da)=\sum_{i=1}^zf(a^i),$$
the right hand expression of course only if the measure is purely atomic.

For such a  population,  in a habitat of carrying capacity $K$, we assume slightly more generally, that individual life and reproduction can be influenced not only by population size $z$ but by the whole array $A$, or some suitable aspect of it.  An individual of age $a$ thus has a random life span with hazard rate $h_A^K(a)$. During life, she gives birth to single children with intensity $b_A^K(a)$ at age $a$. If she dies, she splits into a random number of children which follows a distribution that may depend upon $A$. Its expectation is denoted by $m_A^K(a)$ and the second moment by $v_A^K(a)$, if the mother's age at death was $a$. When the carrying capacity is fixed, we often allow ourselves not to spell it out, writing $h_A(a)$ etc. . However, what we really have in mind remains population-size-dependence, usually but not necessarily measured by the number of individuals alive, and we  allow ourselves the inconsistency sometimes to suffix parameters by $z$ rather than $A$, or even by the ``population density'' $x=z/K$.

 Instead of  the population size $z=|A|=(1,A)$, we could use some other environmental load or crowding measure like $(W,A) = \sum W(a^i)$,  $W(a)$ denoting the body mass or DNA content of a cell aged $a$. Increasing $W$ would correspond to a situation where older cells, being larger, require more of space or resources. 
In the deterministic literature there have been several, more or less {\em ad hoc} choices like linear, exponential or so called von Bartalanffy growth of individual cell mass with age, \cite{webb}, \cite{chapman}. We shall illustrate by linear growth, $W(a)=\kappa + \lambda a$. For a stochastic approach to body mass structured population dynamics cf. \cite{odwyer}.
Generally  the population size could be any (additive) functional of $A$, $(f,A)$.  (Note however that a full-fledged stochastic theory should allow individual variation in the function $W$.)

Whilst there are no deaths and no births, the population changes only by ageing.
When an individual dies its point mass disappears and an offspring number
of point masses at zero age appear. Similarly, when she gives
birth during life a point mass appears at the origin. Thus
population evolution is given by a measure-valued process
$\{A_t^K; t\geq 0\}$. Since process parameters depend upon $K$, there is
a family of such processes indexed by $K$.

In terms of birth and death intensities and mean numbers of children at splitting, the rate of change of population size $(f,A)$ initiated by an $a$-aged member  should then be $f'(a) + f(0)b_A(a)+f(0)h_A(a)m_A(a) -f(a)h_A(a)$, i.e. ageing plus bearing of newborn (zero-aged) children during life plus splitting minus death. The population as measured by $f$  might suitably be termed {\em strictly critical} at population configuration $A$ if the expression vanishes. For an age-independent size measure like the number of individuals around, i.e. $f=1$, $f'(a)=0$, and strict criticality at  $z$ occurs if and only if $b_z(a) + h_z(a)(m_z(a)-1)=0$ for all $a$. Another obvious criticality concept for (classical) population size could be referred to as quenched or frozen criticality: A population is {\em frozen}  critical at size $z$ if and only if $\mu_z(\infty)=1$, where $\mu_z(da)$ denotes the expected reproduction of an individual of age $a$ in a population of size $z$, in terms of intensities and the corresponding life span distribution $L_z$, 
\[\frac{\mu_z(da)}{1-L_z(a)}= (b_z(a) + m_z(a) h_z(a))da.\]
If $b_z(a)=0$ identically in $a$ and $m_z$ is independent of age, the process reduces to a Bellman-Harris age-dependent branching process with population size dependence. Then strict and frozen criticality coincide, and reduce to the classical condition $m_z=1$.  Generally, strict criticality implies frozen.  A third, and fundamental,  concept of criticality is that a population is {\em critical with respect to the age composition A} if and only if
\[(b_A + h_A(m_A-1),A)=0\]
in the case of size being the number of population members, and 
\[(f'+f(0)b_A+h_A(f(0)m_A -f)), A)=0\]
in general. Clearly, strict criticality implies not only frozen but also age composition criticality, as defined. In the next section, we shall see that there is no trend to population change, neither to increase nor to decrease, when the population is composition critical; change is random in the sense that it has a  martingale character.

Under fairly general assumptions, we prove first that a small population either dies out directly, without approaching the carrying capacity, or else comes close to $K$, i.e. reaches any band $[(1-\varepsilon)K,(1+\varepsilon)K], 0< \varepsilon < 1$, in a time of order $\log K$. Once the population size has reached such a level, it stays there for an exponentially long time, i.e. its expected persistence time is $O(e^{cK})$ for some $c>0$. Usually, such results are derived from a Large Deviation Principle yielding the time to exit from a domain of attraction of a fixed point, cf. \cite{FW}. In \cite{lingering}, we gave a  proof by an exponential martingale inequality. Basically, the quick growth and fading of populations follows from a natural principle of monotonicity: if a process remains below (alternatively, above) a certain level during a time period, then it should be larger (smaller) than the similarly started classical branching process, with parameters frozen at the level in question. The long persistence time around the carrying capacity follows from the criticality there.

 In the past, populations have been studied by measure-valued Markov processes with various setups, see eg.  \cite{bb}, \cite{dawson}, \cite{eandk} Section 9.4, and \cite{meleard}, \cite{metivier}, \cite{Oe}, \cite{tran}.
We take the state space to be the
finite positive Borel measures on ${\Bbb R}^+$ with the topology of weak
convergence, i.e.
  $\lim_{n\to\infty}\mu_n= \mu$ if and only if
$\lim_{n\to\infty}(f,\mu_n)=(f,\mu)$ for any bounded and continuous
function $f$ on ${\Bbb R}^+$. M\'etivier \cite{metivier} and Borde-Boussion \cite{bb}
imbedded the space of measures into a weighted Sobolev space.
Oelschl\"ager \cite{Oe} chose the state space as the set of signed measures with
yet another topology.  Our model is closest to Oelschl\"ager's, and the French school around M\'el\'eard, cf. \cite {meleard}, \cite{tran} and other papers but formulated in terms of branching rather than birth-and-death processes. Like our paper, Tran \cite{tran} allows age structure. He considers large populations, scales them, and studies the limit behaviour, obtaining in the case of fixed birth rate, no splitting, and a logistic death rate ($h_z(a) = d(a)+ \eta z$) results on large deviations from a limiting process.

\section{Age  and  population dependent processes.}

The basic tools in our analysis are the
generator of the measure-valued Markov population process and  an integral
representation,  known as Dynkin's formula.
An index $A$ in ${\Bbb P}_A$ and ${\Bbb E}_A$ indicates that the population
started at time $t=0$ not from one newborn ancestor but rather from
$z$ individuals, of ages $A=(a^1,\ldots,a^z)$, respectively.
No index means start from some implicit age configuration. The generator
of a Markov measure-valued population-age-dependent branching
process  was given in \cite{JK}. As mentioned, we switch between general  and population
size dependent parameters with suffixes $A$ and $z$, respectively: $b_z, h_z, m_z$ and $v_z^2$, the latter two being the first and second moments of the offspring at splitting $Y(a)$ of an individual dying at age $a$. As pointed out, reference to the carrying capacity is suppressed, when  $K$ is fixed. The reader can easily replace population size $z=(1,A)$ by any other measure $(f,A)$ of size, total body mass, or environmental impact. 

\begin{theorem}\cite{JK} For a bounded differentiable function $F$ on ${\Bbb R}^+$ and  a  continuously
differentiable  function  $f$ on ${\Bbb R}^+$, the  limit
 \be\label{Dynmain} \lim_{t\to 0}\frac{1}{t}{\Bbb
E}_A\Big\{F((f,A_t))-F((f,A))\Big\}={\mathcal G}F((f,A)), \nonumber \ee
exists, where
\begin{eqnarray}
{\cal G}F((f,A)) =F'((f,A))(f',A)+\sum_{j=1}^zb_z(a^j)\{F(f(0)+(f,A))-F((f,A))\}+\nonumber \\
+\sum_{j=1}^z h_z(a^j)\{{\Bbb
E}_A[F(Y(a^j)f(0)+(f,A)-f(a^j))]-F((f,A))\}. \nonumber \label{MainGen}
\end{eqnarray}
Consequently, Dynkin's formula holds: for a bounded $C^1$ function
$F$ on ${\Bbb R}$ and a $C^1$ function $f$ on ${\Bbb R}^+$

 \be\label{generalDynkin}
F((f,A_t))=F((f,A_0))+\int_0^t {\cal G}F((f,A_s))ds+M^{F,f}_t,
 \ee
where $M^{F,f}_t$ is a local martingale with    predictable
quadratic variation
$$\left<M^{F,f},M^{F,f}\right>_t=\int_0^t {\cal G}F^2((f,A_s))ds-2\int_0^t F ((f,A_s)){\cal G}F ((f,A_s))ds.$$

\end{theorem}
As a corollary the following representation was also obtained; see
\cite{JK}:
\begin{theorem}\label{DynRep} For a $C^1$ function $f$ on $^+$ \be
(f,A_t)=(f,A_0)+\int_0^t(L_{A_s}f,A_s)ds+M_t^f, \label{Dynkin1}, \ee where the   linear operators
$L_A$ are defined by \be
 L_Af=f'-h_Af + f(0)(b_A +h_Am_A),
 \label{GenSD}
\ee and $M_t^f$ is a  local square integrable  martingale with the
sharp bracket given by \be
\left<M^f,M^f\right>_t=\int_0^t\big(f^2(0)b_{A_s}
+f^2(0)v^2_{A_s}h_{A_s}+h_{A_s}f^2-2f(0)m_{A_s}h_{Z_s}f,A_s\big)ds.
\label{sharp} \ee
\end{theorem}
The special choice $f=1$ yields  population size and with $Z_s=(1,A_s)$ the population size at time $s$ and 
\be
Z_t=Z_0 + \int_0^t (b_{Z_s} + h_{Z_s}(m_{Z_s}-1),A_s)ds + M_t^1, \label{popsizeq}\ee
provided parameters are determined by population size.
 
If we turn to body or DNA mass $W,$ it may be noted that it (in principle) increases with age. Criticality therefore requires that the death intensity is mass or age structure dependent, reproduction parameters cannot alone regulate this.  With $W(a) = \kappa + \lambda a, Y_t=(W,A_t)$, and $I(a)=a,$ 
\be
Y_t=Y_0 + \int_0^t (\lambda (Z_s - (Ih_{A_s},A_s))+\kappa(b_{A_s}+h_{A_s}(m_{A_s}-1),A_s)))ds + M_t^W, \label{popageq}\ee
a quite complicated expression, in spite of the straightforward linear growth assumption. It can be slightly simplified in the cell relevant case of splitting, where the bearing term $b_{A_s} =0$. 
 Further, from
Theorem 2.3 of \cite{JK} it follows that if $f\ge 0$ satisfies the
(linear growth) condition (H1)
  $$
  |(L_Af, A)|\le C(1+(f,A))\eqno (H1)
  $$
for some $C>0$ and any $A$, and if $(f,A_0)$ is integrable, then so is
$(f,A_t)$. Its expectation is bounded by
\begin{equation}\label{meanbound}
{\Bbb E}[(f,A_t)]\le \Big({\Bbb E}[(f,A_0)]+Ct\Big)\Big(1+e^{Ct}/C\Big).
\end{equation}
We use expectation without an index to indicate that the starting age configuration $A_0$ is arbitrary and may well be random.   $C$ denotes a generic constant, not necessarily the same in different contexts. A family of functions $f_z$ is called uniformly bounded if $\sup_{z,a}|f_z(a)|<\infty$.
The following corollary is easy to check.

\begin{corollary} Suppose that the functions $b_z$,
$m_z$, and $h_z$ are uniformly bounded.   If  $f$ and $f'$
are bounded, then the growth condition $(H1)$ is satisfied and the
conclusion (\ref{meanbound}) holds.

In particular, the function $f(a)=1$ satisfies $(H1)$ and so
\begin{equation}\label{meansize}
{\Bbb E}[Z_t]\le \Big({\Bbb E}[Z_0]+Ct\Big)\Big(1+e^{Ct}/C\Big).
\end{equation}
Further, if  the functions $v_z$ are uniformly bounded as well, then $M^1_t$ is a square integrable martingale with the quadratic variation
\begin{equation}\label{boundQV}
\left<M^1,M^1\right>_t\le C\int_0^tZ_sds.
\end{equation}
\end{corollary}
Similarly, it can be directly checked that the linear weight growth function $W$ satisfies $(H1)$, so that a bounding inequality corresponding to (\ref{meansize}) holds. If there is a maximal age, so that $W$ is bounded, the quadratic variation can also be bounded.

\section{Extinction or growth}
From now on we consider population size, intepreted as population number, and population size dependence of demographic parameters, $h_z,b_z,$ etc. We say that reproduction {\em decreases with  population size} if  for all $t,z$
\be
Z_u \leq z, u\leq t \Rightarrow Z_u \stackrel{d}{\geq} \tilde{Z}_u, u\leq t,
\ee
where $\{\tilde{Z_t}\}$ is the process with parameters frozen at $z$ and the same starting conditions as  $\{Z_t\}$.

Following \cite{lingering}, consider a general population-size dependent branching process with a fixed carrying capacity $K$, as described. Let it start at time $t=0$ from $z$ individuals. To ease notation we take them all as newborn. Such a population must die out eventually \cite{pjstab}. What are then chances that it will reach a size in the vicinity of the carrying capacity, before extinction? We write $T$ for the time to extinction and $T_d$ for the time when the population first attains a size $\geq dK, 0<d<1, z<dK$,letting it equal infinity if this never occurs. Clearly,
\[T < T_d \Rightarrow \forall t, Z_t < dK.\]
Let units with a tilde denote entities pertaining to a population-size independent branching process with the fixed parameters $b_{dK},h_{dK},m_{dK}$. Then, if reproduction decreases with population size, 
\[{\Bbb P}(T<T_d)\leq {\Bbb P}(\tilde{T} < \infty) =\tilde{q}^z,\]
where $\tilde q$ is the extinction probability of the population-size independent branching process. If $\tilde{m}_d>1$ and $\tilde{\sigma}^2_d$ denote the mean and variance of the total reproduction of an individual in this latter process, we have by Haldane's inequality (\cite{haccou}, p. 125) that the probability of the original population never reaching $dK$ is
\[{\Bbb P}(T <T_d) \leq \left(1-\frac{2(\tilde{m}_d -1)}{\tilde{\sigma}^2_d +\tilde{m}_d(\tilde{m}_d-1)}\right)^z.\]
With a positive chance, the population will thus reach a size of order $K$. Since it grows quicker than the process $\tilde{Z}_t$ while under the level $dK$, and the latter process grows exponentially, we can conclude that attaining $dK$ will occur within a time of order $\log K$.
\begin{theorem} If reproduction decreases with population size and the population is frozen critical at size $K$, then any population size $dK, 0<d<1$ is attained  with positive probability within a time $T_d =O(\log K)$, as $K\to\infty$
\end{theorem}

\section{An era of stable size}
We proceed to see that once the population size has reached  the interval $[K-\varepsilon K,K+\varepsilon K]$ it remains there an exponentially long period.
In other words, in terms of the process scaled by $K$, it
takes exponentially long  to exit from $[1-\varepsilon
,1+\varepsilon]$. The property ensuring this lingering around the carrying capacity is that the population reproduces subcritically above level $K$, supercritically below, and
critically at $K$. As we shall see, this forces the scaled population size to converge to one.



Criticality is understood in the strict sense. Define the criticality function
\be
\chi_z= b_z+h_z(m_z-1)=L_z1,
\ee
in terms of the operator $L$ (\ref{GenSD}). Then criticality means that
$\chi_z(a) = 0$ for all $a$, as soon as $z=K$. The superscript $K$ is usually not spelled out and somewhat carelessly we switch between  dependence on population size and on the scaled population size (``density'') $x=z/K$, writing $\chi_x$, so that $\chi_1 = 0$.
In this density notation, assume that $\chi$ satisfies a Lipschitz condition in the neighbourhood of 1:

{\bf Assumption 1} There is a constant $C$ such that
\be |\chi_x|=|\chi_x-\chi_1|\le C|x-1|.\label{lipschitz} \ee

\begin{theorem} \cite{lingering} Write  $X_t^K=Z_t^K/K$ for the population density and suppose that $X_0^K\to 1$ in probability, as $K\to\infty$, and that Assumption 1 holds. Then $X_t^K$ converges in probability to 1, uniformly on any time interval $[0,T],T>0$. In other words, for any $\eta>0$
$$\lim_{K\to\infty}{\Bbb P}(\sup_{t\le T}|X_t^K-1|>\eta)=0.$$
\end{theorem}
\begin{proof}
By Equation \eqref{popsizeq} \be \label{mainspec}
X_t^K=X_0^K+\int_0^t\Big(\chi_{X_s^K},\frac{1}{K}A^K_s\Big)ds+\frac{1}{K}M_t^{1,K}.\ee
In terms of the reproduction variance $\sigma^2_x= v_x -m^2_x$, the martingale quadratic variation reduces to
$$\left<\frac{1}{K}M^{1,K},\frac{1}{K}M^{1,K}\right>_t= \frac{1}{K}\int_0^t
\big((b_{X^K_s}+\sigma^2_{X^K_s}+(m_{X^K_s}-1)^2)h_{X^K_s},\frac{1}{K}A^K_s\big)ds.
$$

First we show that \be\label{limsup}\limsup_{K\to\infty}\sup_{t\le
T}X^K_t\le e^{CT} \;\;\mbox{in probability}.\ee

Indeed, since the criticality function is bounded by some $C>0$, and $(1,A_s^K/K)=X^K_s\le \sup_{u\le s}X^K_u$,
\begin{eqnarray}
\sup_{t\le T}X^K_t&\le&
X_0^K+\frac{1}{K}\int_0^T|(\chi_{X^K_s}1,A^K_s)|ds+\sup_{t\le
T}\frac{1}{K}|M_t^{1,K}|\nonumber\\
&\le& X_0^K+\frac{1}{K}\int_0^TC(1,A^K_s)ds+\sup_{t\le
T}\frac{1}{K}|M_t^{1,K}| \nonumber\\
&\le& X_0^K+\sup_{t\le T}\frac{1}{K}|M_t^{1,K}|+C\int_0^T\sup_{u\le s}X^K_uds.
\nonumber
\end{eqnarray}
Gronwall's inequality in the form given in \cite{fkbook} p. 21 yields
\be\label{supbound}
 \sup_{t\le T}X^K_t\le  (X_0^K+\sup_{t\le T}\frac{1}{K}|M_t^{1,K}|)e^{CT}.\ee

By Doob's inequality,
\begin{eqnarray}
{\Bbb P}(\sup_{t\le T}\frac{1}{K}|M_t^{1,K}|>a)&\le&
\frac{1}{a^2}{\Bbb E}\left<\frac{1}{K}M^{1,K},\frac{1}{K}M^{1,K}\right>_T\nonumber\\
&\le&\frac{C}{a^2K}\int_0^T{\Bbb E}X_t^Kdt\le \frac{C_1{\Bbb E}X_0^K}{a^2K}\to 0,
\nonumber\end{eqnarray}
where the bound from \eqref{boundQV} was used.
Hence the quadratic variation
$$\left<\frac{1}{K}M^{1,K},\frac{1}{K}M^{1,K}\right>_T\to 0$$, as $K\to\infty$,
and
$$\sup_{t\le T}\frac{1}{K}M_t^{1,K}\stackrel{p}{\to} 0\;\;\mbox{, as}\;\;
K\to\infty .$$
The claim
\eqref{limsup} follows from \eqref{supbound}.

Now we prove the main assertion. From Corollary 1,

 \begin{eqnarray}
\sup_{t\le T}|X_t^K-1|&\leq &|X_0^K-1|+
 \int_0^T\Big|\Big(\chi_{X^K_s},\frac{1}{K}A^K_s\Big)\Big|ds +
\sup_{t\le T}\frac{1}{K}|M_t^{1,K}|\nonumber\\
&\le&|X_0^K-1| + \sup_{t\le T}\frac{1}{K}|M_t^{1,K}|+C\sup_{t\le T}X_t^K
 \int_0^T \sup_{u\le s}|X_s^K-1 |ds,\nonumber
 \end{eqnarray}
thanks to the Lipschitz condition \eqref{lipschitz}. Gronwall's
inequality is there again to conclude that
 \begin{equation}
\sup_{t\le T}|X_t^K-1|  \le  \Big(|X_0^K-1| + \sup_{t\le
T}\frac{1}{K}|M_t^[1,K]|\Big)e^{C\sup_{t\le T}X_t^K}.\nonumber
 \end{equation}
 The first term converges to 0 by assumption and we saw that so does the second. Relation \eqref{limsup}
 completes the proof.

\end{proof}

 An exponential bound on the exit time from the vicinity of $K$ requires
exponential moments of the process. Hence,  we assume
that the offspring distributions have exponential moments which are
bounded. Then the  process $Z_t^K=(1,A_t^K)$ has
exponential moments as well. Let $\phi_A(t)(a)=\Bbb{E}_A[e^{tY(a)}]$ denote
the conditional moment generating function given $A$ of the number $Y(a)$ of offspring at death of an $a$-aged individual splitting in a population with age composition $A$. Similarly, $\Bbb{P}_A$ denotes  offspring probabilities in a population of size and composition $A$.

The following condition (\cite{tran}, Assumption 3) may seem strange at first sight, but it serves to give the process subcriticality above the carrying capacity a strict form. In it $\phi_A(t)$ denotes the function $\phi_A(t)(\cdot)$.

{\bf Assumption 2}.
For any $K$ there exists a population size $V_K >K$ such that
\begin{eqnarray}\label{VK}
 (e^{1/K}-1)b_{A}+
 (\phi_A(1/K)e^{-1/K}-1) h_{A}&\le& 0,\;\;  \mbox{ whenever }
 (1,A)>V_K,\nonumber
\end{eqnarray}
and $V_K/K$ is bounded for large $K$. 

Since the reproduction is subcritical for population sizes larger than $K$, such a  number exists. 
Indeed, for large $K$ 
$$(e^{1/K}-1)b_{A}+
 (\phi_A(1/K)e^{-1/K}-1) h_{A}\sim \frac{1}{K}\Big(b_{A}+
 (m_A-1 )h_{A}\Big),$$
which is negative for large $(1,A)$. The assumption needed is that this occurs not too far away from $K$, when also the latter is large.
An example is provided by the binary splitting with $b=0$ and $Y=Y(a)$ independently of age at split, mentioned in the Introduction and further explored in \cite{ksvhj}:
 ${\Bbb P}_A(Y=2)=K/(K+z), z= (1,A)$, then $V_K$ is determined from
 $$\frac{z}{K+z}e^{-1/K}+\frac{K}{K+z}e^{1/K}={\Bbb E}_ze^{(Y-1)/K}=1.$$
Solving in $z$ gives  $V_K=e^{1/K}K$.

\begin{theorem} Let $X^K_t$ be the  population size scaled by the carrying capacity $K$. Suppose that exponential moments of the offspring number at split exist and that Assumptions 1 and 2 are in force.
Then, there is a constant $C$, independent of $K$, such that for any $t$
 \be\label{expMom}
 {\Bbb E}[e^{X_t^K}]\le {\Bbb E}[e^{X_0^K}]e^{Ct}.
\ee
\end{theorem}

\begin{proof} Since we consider a process for fixed $K$, dependence upon the latter is suppressed in notation. The statement follows by taking $F$ in  \eqref{generalDynkin} as the exponential function, or rather, to be precise, letting it equal smooth bounded functions that agree with the exponential on bounded intervals and a localizing sequence  $T_n = \inf \{t: Z_t>n\}$.

With
 $x=z/K$ we have
\begin{eqnarray*}
{\cal G}(F(1,A))&= &F'((f,A))(f',A)+\sum_{j=1}^zb_A(a^j)\{F(f(0)+(f,A))-F((f,A))\}\\
&&+\sum_{j=1}^z h_A(a^j)\{{\Bbb E}[F(Y_Af(0)+(f,A)-f(a^j))]-F((f,A))\}\\
&= &(e^{x+1/K}-e^{x})\sum_{j=1}^zb_A(a^j)  +\sum_{j=1}^z
h_A(a^j)\Big({\Bbb E}[e^{x+(Y_A-1)/K}]-e^{x}\Big)\\
&= &e^{x}\Big((e^{1/K}-1)(b_A,A)+({\Bbb E}[e^{(Y_{A}-1)/K}]-1)
(h_{A},A)\Big).
\end{eqnarray*}

Hence we obtain by \eqref{generalDynkin}

\be e^{X_t}= e^{X_0}+\int_0^t e^{X_s} \Big((e^{1/K}-1)b_{A_s}+
 ({\Bbb E}[e^{(Y_{A_s}-1)/K}]-1) h_{A_s},A_s\Big)ds+M^{\exp}_t,\ee where $M^{\exp}_t$ is a local martingale.
Localizing and taking expectation,

$$ {\Bbb E}[e^{X_{t\wedge T_n}}]={\Bbb E}[e^{X_0}]+{\Bbb E}[\int_0^{t\wedge T_n} e^{X_s}
\Big((e^{1/K}-1)b_{A_s}+
 ({\Bbb E}[e^{(Y_{A_s}-1)/K}]-1) h_{A_s},A_s\Big)ds].$$

Now we use that the reproduction is subcritical above $K$, that
the parameters $b_A$ and $h_A$ are bounded, and that the function
under the integral is negative for values of $Z_s>V_K$ or $X_s>V_K/K$.
For $Z_s<V_K$, the inequalities
$(e^{1/K}-1\le C/K$ and $|\phi_A^K(1/K)e^{-1/K}-1|\le C/K$ show that the
integrand does not exceed $C V_K/K$.
 \begin{eqnarray*}
{\Bbb E}[e^{X_{t\wedge T_n}}]&\le& {\Bbb E}[e^{X_0}]+C\frac{V_K}{K}{\Bbb E}[\int_0^t
e^{X_{s\wedge T_n}} I(Z_s\le V_K)ds]\\
&\le& {\Bbb E}[e^{X_0}]+C{\Bbb E}[\int_0^t e^{X_{s\wedge T_n}}ds],
 \end{eqnarray*}
 where $C$ is a constant independent of $K$, since $V_K/K$ is assumed bounded.
Gronwall's inequality yields
\begin{equation}
{\Bbb E}[e^{X_{t\wedge T_n}}]\le {\Bbb E}[e^{Z_0}]e^{CT},
\end{equation}
where $C$ does not depend on $K$.

Letting $n\to\infty$, we have obtained  \ref{expMom}.

\end{proof}

The main result on persistence time is from \cite{lingering}.

\begin{theorem}\label{exptime} Assume that
$X_0^K\to 1$ in probability. For any $\varepsilon > 0$, let $\tau^K=\inf \{t:
|X^K_t-1|>\varepsilon \}$. Suppose that the previous assumptions
hold and also that the number of offspring through splitting at death is
bounded by some constant. Then ${\Bbb E}[\tau^K]$ is exponentially
large in $K$, i.e. for some positive constants $C,c$
$${\Bbb E}[\tau^K]> Ce^{cK}.$$
\end{theorem}

\begin{proof}

We start from equation \eqref{mainspec} for $X_t^K$ and recall from
Corollary 2 that the predictable quadratic variation of the
martingale is bounded,
\be\label{QVintegralbound}\left<\frac{1}{K}M^1,\frac{1}{K}M^1\right>)_t\le
\frac{C}{K}\int_0^t X_s^Kds \ee since the parameter functions are
uniformly bounded.
%

As we need exponential moment bound for the integral
$\int_0^1 X_s^Kds$, we shall use the following inequality, obtained
by Jensen's inequality for the uniform distribution on $[0,1]$ combined with
an exponential function: for any integrable function $g$ on
$[0,1]$. \be\label{JensenInt}
 \int_0^1e^{g(s)}ds\ge e^{\int_0^1 g(s) ds}.
\ee

By the bound \eqref{expMom} for the exponential moment ${\Bbb
E}e^{{X^K_s}}$, \be\label{expbound} {\Bbb E} e^{\int_0^1X^K_s ds}\le
 {\Bbb E}\int_0^1e^{X^K_s}ds\le Ce^{C},
\ee where the last bound is independent of $K$.

Next,  we establish an exponential bound for the probability of exit
up to time 1, when the normed population started at $x\in(1-\eta,1+\eta)$ for
$\eta<\varepsilon/6$.
$${\Bbb P}_x(\tau\le 1)={\Bbb P}_x(\sup_{t\le 1}|X^K_t-1|>\varepsilon).$$
Denote
$$\int_0^t\Big(b_{X^K_s}+(m_{X^K_s}-1)h_{X^K_s},\frac{1}{K}A^K_s\Big)ds=I^K_t.$$
Then, since $X_0^K=x$, and $|X^K_t-1|\le |X^K_t-x|+|x-1|\le \eta+
|X^K_t-x|$
\begin{eqnarray*}
{\Bbb P}_x(\sup_{t\le 1}|X^K_t-1|>\varepsilon)&\le& {\Bbb
P}_x(\sup_{t\le 1}|I^K_t+\eta|>\varepsilon/2)+ {\Bbb P}_x(\sup_{t\le
1}|\frac{1}{K}M^1_t|>\varepsilon/2)\\
&\le& {\Bbb P}_x(\sup_{t\le 1}|I^K_t|>2/3\varepsilon)+ {\Bbb
P}_x(\sup_{t\le 1}|\frac{1}{K}M^1_t|>\varepsilon/2).
\end{eqnarray*}

As $X^K_0\to 1$, $X^K_t\to 1$ by Theorem 4, and consequently $\chi_{X_t^K}\to
\chi_1=0$, all convergences taking place in probability. Hence $\Big (\chi_{X^K_s}, A^K_s\Big)=o(K)$, and we have
by the exponential form of Chebyshev's inequality that
$${\Bbb P}_x(\sup_{t\le 1}|I^K_t|>2/3\varepsilon)\le {\Bbb P}_x\left( \int_0^1|\Big (\chi_{X^K_s}, A^K_s\Big)|ds>2/3\varepsilon K\right)$$
$$\le e^{-2/3\varepsilon K}{\Bbb E}_xe^{  \int_0^1|  (\chi_{X^K_s}, A^K_s )|ds   }=e^{-2/3\varepsilon K+o(K)}\le e^{-CK}$$
for some $C$.

The second probability ${\Bbb P}(\sup_{t\le
1}|M^1_t|>K\varepsilon/2)$ is controlled by an exponential
martingale inequality due to Chigansky and Liptser, see \cite{CL} Lemma 4.2. \be
{\Bbb P}_x(\sup_{t\le T}|M_t^1|>\varepsilon,
\left<M^1,M^1\right>_T\le q)\le
2e^{-\frac{\varepsilon^2}{k\varepsilon +q}}, \ee where $k$ is a
bound on the jumps of $M^1$. In our case $k=B$, if $B$ is the
maximal number of children at splitting. Hence replacing
$\varepsilon$ by $K\varepsilon/2$
$${\Bbb P}_x(\sup_{t\le 1}\frac{1}{K}|M^1_t|>\varepsilon/2)\le 2e^{-\frac{\varepsilon^2K^2/2}{BK\varepsilon +2q}}
+{\Bbb P}_x\left(\left<\frac{1}{K}M^1,\frac{1}{K}M^1\right>_1>
q\right)\le Ce^{-c(\varepsilon)K}, $$ where we used the bound \eqref{QVintegralbound}
on quadratic variation, and Chebyshev's
inequality with the exponential moments \eqref{expbound},

$${\Bbb P}_x\left(\left<\frac{1}{K}M^1,\frac{1}{K}M^1\right>_1>
q\right)\le {\Bbb P}_x\left(\frac{C}{K}\int_0^1 X_s^Kds >q\right)$$
$$\le
e^{-CK}{\Bbb E}_xe^{\int_0^1 X_s^Kds}\le Ce^{-CK} \int_0^1 {\Bbb
E}_xe^{X_s^K }ds\le Ce^{-CK},$$ where the last inequality is by
\eqref{expbound}. The final step is a recursive argument, formulated
in terms of the filtration $\{{\cal F}_n:=\sigma(\{A^K_t, t\le
n\})\}$:
\begin{eqnarray*}
{\Bbb P}_0(\tau^K>n)={\Bbb P}(\sup_{t\le n}|X^K_t-1|<\varepsilon)\\
={\Bbb P}_0(\sup_{t\le n-1}|X^K_t-1|<\varepsilon,\sup_{n-1\le t\le n}|X^K_t-1|<\varepsilon)\\
\ge{\Bbb P}_0\Big(\sup_{t\le n-1}|X^K_t-1|<\varepsilon,\sup_{n-1\le t\le n}|X^K_t-1|<\varepsilon, |X_{n-1}^K-1|<\eta\Big)\\
={\Bbb E}_0\Big( {\Bbb P}\big((|X_{n-1}^K-1|<\eta,\sup_{n-1\le t\le n}|X^K_t-1|<\varepsilon)|{\cal F}_{n-1}\big); \sup_{t\le n-1}|X^K_t-1|<\varepsilon)\Big)\\
\ge \inf_{x\in (1-\eta,1+\eta)} {\Bbb P}_x(\sup_{t\le 1}|X^K_t-1|<\varepsilon){\Bbb P}_0(\tau>n-1)\\
\ge \Big(\inf_{x\in (1-\eta,1+\eta)} {\Bbb P}_x(\sup_{t\le 1}|X^K_t-1|<\varepsilon)\Big)^n.
\end{eqnarray*}
So
$${\Bbb E}_0\tau^K >\sum_n {\Bbb P}_0(\tau^K>n)>\sum_n \Big(\inf_{x\in (1-\eta,1+\eta)} {\Bbb P}_x(\sup_{t\le 1}|X^K_t-1|<\varepsilon)\Big)^n$$
$$=\frac{1}{1-\inf_{x\in (1-\eta,1+\eta)} {\Bbb P}_x(\sup_{t\le 1}|X^K_t-1|<\varepsilon)}$$
$$=\frac{1}{\sup_{x\in (1-\eta,1+\eta)} {\Bbb P}_x(\sup_{t\le 1}|X^K_t-1|>\varepsilon)}>Ce^{cK}.$$

\end{proof}

\section{A simple example with far-reaching conclusions}
The main drawback of results like Theorem \ref{exptime}, and generally large deviations principles, is the implicitness of the constants involved. In the ``bare bones'' binary splitting case \cite{ksvhj},  the situation is more transparent. Indeed, 
let each individual live for one season (=generation), begetting two children in the next with probability $K/(K+z)$, and none otherwise, if the population size is $z$. Then,  the constants $C$ and $c$ in Theorem \ref{exptime} can be chosen as 1 and 
\begin{equation}
c=\frac{d(1-d)^2}{8(1+d))} \label{expconstant},
\end{equation}
for any $0<d<1$ and any starting population size  $z\geq dK$.
(This corrects a misprint in the statement, without proof, of this result in  \cite{ksvhj}.) 

If $d=0.5$, say, then $c = 0.01$. Thus even a population in a habitat  with a biologically small carrying capacity of say one thousand individuals, will probably persist for many generations, $e^{11.25} \approx 20.000$, if it does not die out during the first few rounds. Only with very short generation times, like one hour,  for certain cells or bacteria, this will be of the magnitude a couple of years. Simulations further indicate that this approximation is excellent: biologically small carrying capacities may well be mathematically large. 

By analogy, endangered more longlived species, say of a size of a couple of thousand individuals,  seem not threatened by demographic stochastic fluctuations during time periods of human scales. This certainly corresponds to established beliefs, but further investigation may not be unwarranted.  Anyhow, the conclusion at this stage is that real dangers rather lie in trends or in varying environments, not included in this type of models, describing unvarying carrying capacities.  We intend to study randomly varying carrying capacities in a sequel paper. In the present context more pertinent questions would seem to concern  population properties during the long lingering around high carrying capacities.

But first the proof (due to V. A. Vatutin) of \eqref{expconstant}.

Define $\tau^K$ as the hitting time of $dK$, in slight disagreement with earlier notation.

\begin{theorem}
For any $K>0, n\in N$,  and $z\ge dK, 0<d<1 $,
$${\Bbb P}_z(Z_1> dK)\ge 1-e^{-cK}.$$ 
$${\Bbb P}_z(\tau^K>n)\ge (1-e^{-cK})^n,$$ and
$${\Bbb E}_z[\tau^K]\ge e^{cK}$$
with $c$ as defined in \eqref{expconstant}.
\end{theorem}

\begin{proof}
 Janson's inequality for binomial distributions \cite{Janson}  tells that a binomial random variable
$X$, with parameters $k,p$, satisfies
$${\Bbb P} (X\le kp-a)\le e^{-a^2/(2kp)},$$  for any $a>0$.
Of course, the bound  remains correct when the inequality $X\le kp-a$ is multiplied by two. But with  $p(z) = K/(K+z)$ and  $Z_0=z, \ Z_1\sim 2\mbox{Bin}(z, p(z))$, and so
$${\Bbb P}_z(Z_1\le dK)={\Bbb P}_z(Z_1\le 2zp(z)-2a)$$
$$\le \exp\Big(-\frac{(2zp(z)-dK)^2}{8zp(z)}\Big)=\exp\Big(-K\frac{(f(x)-d)^2}{4f(x)}\Big),$$
where $x=z/K$, $f(x)=2xp(z)$, and the constant $a$ of Janson's theorem is $(2zp(z)-dK)/2$. This is ok, since for $1>d, x>d$,  $f(x)>f(d)>d$ , and therefore
$a>0$. Further, $(u-d)^2/(4u)$ is an  increasing function of $u>d$ so that for
$x>d$
$$\frac{(f(x)-d)^2}{4f(x)}>\frac{(f(d)-d)^2}{4f(d)}=c.$$

Hence, for any $x>d$, i.e. $z>dK$,
$${\Bbb P}_z(Z_1\le dK)\le e^{-cK},$$ as claimed.

For the second assertion,
$${\Bbb P}_z(\tau^K>1) = {\Bbb P}_z(Z_1> dK) \ge 1-e^{-cK}.$$

 We use induction to show that for any $z\ge dK$ and natural number $n$
$${\Bbb P}_z(\tau^K>n)>(1-e^{-cK})^n.$$ 
By the Markov property,
$${\Bbb P}_z(\tau^K>n+1)={\Bbb P}_z(Z_1> dK,\ldots,Z_{n+1}> dK)$$
$$=\sum_{k> dK}{\Bbb P}_z(Z_1=k,Z_2> dK,\ldots,Z_{n+1}> dK)$$
$$=\sum_{k> dK}{\Bbb P}(Z_2> dK,\ldots,Z_{n+1}> dK|Z_1=k) {\Bbb P}_z(Z_1=k)$$
$$=\sum_{k> dK}{\Bbb P}_k(\tau^K>n) {\Bbb P}_z(Z_1=k).$$
 Induction yields that this is
$$\ge (1-e^{-cK})^n {\Bbb P}_z(Z_1> dK)\ge (1-e^{-cK})^{n+1}, $$
as required. The last assertion follows from the relation ${\Bbb E}[\tau^K]=\sum_n
{\Bbb P} (\tau^K>n).$
\end{proof}

There is a corresponding assertion for the waiting time until population size leaves a band around the carrying capacity upwards:

\begin{theorem}
For any  $d>1$ write $c_1=\frac{(d-1)^2}{8(d+1)}$. Then for any $K$
and $z\le dK$,
$${\Bbb P}_z(Z_1< dK)\ge 1-e^{-c_1K}.$$ Moreover,
for any $z\le dK$
$${\Bbb P}_z(\tau^K>n)>(1-e^{-c_1K})^n,$$ and
$${\Bbb E}_z[\tau^K]>e^{c_1K}.$$
\end{theorem}
\begin{proof}
 Since $k-\mbox{Bin}(k,p)\stackrel{d}{=} \mbox{Bin}(k,1-p)$,
$${\Bbb P}(\mbox{Bin}(k,p)\ge kp+a)={\Bbb P}(k-Bin(k,p)\le k(1-p)-a)$$
$$={\Bbb P}(\mbox{Bin}(k,1-p)\le k(1-p)-a)\le e^{-a^2/(2k(1-p))}.$$
Again, the bound  remains after multiplication of the inequality
 by 2. Thus, with $Z_0=z$, and $2a=dK-2zp(z)$ and $f$ as above,

$${\Bbb P}_z(Z_1\ge dK)={\Bbb P}(2\mbox{Bin}(z,p(z))\ge dK)$$
$$={\Bbb P}(2\mbox{Bin}(z,p(z))\ge 2zp(z)+ dK-2zp(z)).$$
Now $a>0$ is equivalent to $d-f(x)>0$. But since $f$ increases and
 $x<d$, $f(x)<f(d)$, and so $f(d)<d$
as $d>1$ now. Hence,

$${\Bbb P}_z(Z_1\ge dK)\le \exp\Big(-\frac{(dK-2zp(z))^2}{8z(1-p(z))}\Big)=\exp\Big(-K\frac{(d-f(x))^2}{4(2x -f(x))}\Big).$$

Now, the smallest value of $d-f(x)$ provided $x<d$ is
$d-f(d)$ by the monotonicity of $f$. The function
$2x-f(x)=\frac{2x^2}{1+x}$ is positive and   increasing. Therefore its largest value in $(0,d)$ is
$2d-f(d)$. Hence
$$\frac{(d-f(x))^2}{4(2x -f(x))}\ge \frac{(d-f(d))^2}{4(2d -f(d))}=\frac{(d-1)^2}{8(d+1)}=c_1,$$
and
$${\Bbb P}_z(Z_1\ge dK)\le e^{-c_1K}.$$
The rest follows as in the preceding theorem.

\end{proof}

In our context it is leaving downwards that is crucial. It may however be worth noting that if $d=1\pm \epsilon$, with $0<\epsilon <1$, then $c<c_1$ works for both cases. 

\section{Stabilisation of the population composition}

Thus, we turn to the  long period  of lingering around the carrying capacity.  Will the population composition have the time to stabilise, and then how can the pseudo-stable age-distribution and other aspects of the composition be described? The age distribution was recently investigated in \cite{trudy}, resulting in the following two main theorems, on tightness, and convergence, respectively. The convergence of the mass distribution (for decent $W$) follows trivially, cf. \cite{odwyer}.

\begin{theorem}\label{tight}
Assume that all demographic parameters are uniformly bounded. Suppose also that the support of $\bar A^K_0$
and its total mass are bounded, $\sup_K\inf\{t>0: A^K_0((t,+\infty))=0\}<\infty$ and $D=\sup_K
|\bar A^K_0|<\infty$.  Then, the family  $ \{\bar A^K_t, t\ge 0\}_{K}$ is tight in ${\Bbb D}({\Bbb R}^+, {\cal M}({\Bbb R}^+))$.
\end{theorem}

The proof hinges upon Jakubowski's criteria for weak convergence of random measures in spaces again with weak topology,  Theorem 4.6 of
  \cite{Jakub}: 
  A sequence $\mu^K$ of ${\Bbb D}({\Bbb R}^+, {\cal M}({\Bbb R}^+))$-valued random
elements  is tight if and only if the following two conditions are
satisfied.

 J1. (Compact Containment) For each $T > 0$ and $\eta>0$ there exists a compact set
$\mathcal{C}_{T,\eta}\subset {\cal M}({\Bbb R}^+)$ such that
$$\liminf_{K\to\infty}{\Bbb P}(\mu^K_t\in \mathcal{C}_{T,\eta} \;\forall t \in [0, T])>1-\eta.$$

J2. (Separable Coordinate Tightness) There exists a family ${\Bbb F}$ of
real continuous functions $F$ on ${\cal M}({\Bbb R}^+)$ which separates
points in ${\cal M}({\Bbb R}^+)$, is closed under addition, and such that
for every $F\in {\Bbb F}$, the sequence $\{F(\mu^K_t), t\geq0\}_K$
is tight in ${\Bbb D} ({\Bbb R}^+,{\Bbb R})$.

We refer to \cite{trudy} for the technical verification of them in our circumstances, and also for the proof of the convergence theorem that follows, with the help of a suitable smoothness concept.
\begin{definition}
A population process will be said to be {\em demographically smoothly
  density dependent}, or for short just demographically smooth,  if:
\begin{itemize}
\item[C0] The model parameters, $b,h,m$, are uniformly bounded.
\item[C1] They are also normed uniformly Lipschitz in the following
  sense: there is a $C>0$ such that for all $u$ and $K$, $\rho(\mu,\nu)$
  denoting the Levy-Prokhorov distance between measures $\mu$ and $\nu$,
\begin{itemize}
\item $|b^K_A(u)-b^K_B(u)|\le C\rho(A/K,B/K)$,
\item $|h^K_A(u)-h^K_B(u)|\le C\rho(A/K,B/K)$,
\item $|m^K_A(u)-m^K_B(u)|\le C\rho(A/K,B/K)$.
\end{itemize}
\item[C2] $\bar A^K_0 \Rightarrow \bar A_0^\infty$, and $\sup_K|\bar A_0^K|<\infty$. We say that the process {\em stabilises initially}.
\end{itemize}

\end{definition}

\begin{theorem}\label{agedistribution} In a demographically smoothly density dependent
  population process, the processes
$\bar A^K$ converge weakly in the Skorokhod space ${\Bbb D} ({\Bbb R}^+, {\cal
M}({\Bbb R}^+))$. The limiting measure-valued process,  $\bar A^\infty$, displays no
randomness, and for any test function $f$, $\bar A^\infty$ satisfies the integral equation
\begin{equation}
(f,A_t)=(f,\bar A_0^\infty)+\int_0^t(L^\infty_{A_s}f,A_s)ds, \label{limit}
\end{equation}
  where $A_t$ is short for  $\bar A_t^\infty$ and 
  $$
 L^\infty_Af=f'-h^\infty_Af + f(0)(b^\infty_A +h^\infty_Am^\infty_A).$$
\end{theorem}

\begin{remark} Equation \eqref{limit} is the weak form of the classical
  McKendrick-von Foerster equation for the   density of $A_t$, $a(t,u)$
$$(\frac{\partial  }{\partial t}+\frac{\partial }{\partial
u})a(t,u)=-a(t,u)h_{A_t}(u),\;\;a(t,0)=\int_0^\infty
m_{A_t}(u)h_{A_t}(u)a(t,u)du.$$
It can be obtained by integration by parts and the adjoint operator $L^*$,
$$(f,A_t)=(f,A_0)+\int_0^t(f,L^*_{A_s}A_s)ds.$$
Of course, smoothness of the density must also  be proved. The derivation of the equation  in the present context further underpins prevailing deterministic theory. REF?
\end{remark}

Once the unique existence of a time limiting age distribution with a
density has been established, its form follows in the usual manner from
the transport equation above, letting $t\to\infty$. The
derivative of $a$ with respect to time vanishes in the limit, and
$h_t(u)\to$ some $h(u)$, to be inserted in the limiting equation.

It is important to note, though, that the latter has a trivial null
solution, if the population starts from a bounded number of ancestors,
a mutant or a limited number of invaders, so that $\lim_{K\to\infty}
|\bar{A}_0^K| = 0$. As we have seen, such populations either grow to reach a band
around the carrying capacity or die out before that. It is an interesting task to describe the asymptotic age distribtion at time $T^K_d$, as $K\to \infty$, and then the evolution
of the process in a suitable evolutionary time scale that starts when
the population enters a band around the carrying capacity, provided it
so does. 

\section{The time of decay}

The last stage of a population's existence is that when it left a band around the carrying capacity, never to return. Its duration is the time to extinction $T$, from a level $x=aK$, given that the maximum of the process $\bar{Z}$ will never exceed $y=bK, 0<a<b<1$, and $b$ is suitably chosen so as to avoid excessive random overshooting. Since the process is supercritical below the carrying capacity, conditioning upon a maximum value being less than $K$ implies extinction of the various concerned branching processes with frozen parameters,  all supercritical. Since supercritical general branching processe, conditioned to die out, are subcritical \cite{lageras}, one should expect a behaviour in line with the path to extinction of large subcritical processes. For those the survival time of a $K$-sized population, $K\to\infty$,  is of the form, cf. \cite{pnas}, 
\[T = (\log K - c +\eta + o(1))/|\alpha|,\] 
where  $c>0$ is a constant, $\eta$ a Gumbel distributed random variable, and $\alpha$ the (negative) Malthusian parameter of the subcritical process.

In the present case, there is no well defined Malthusian parameter, since rates vary with population size. Furthermore, extinction is only guaranteed if extinction probabilities of the frozen processes involved stay away from zero \cite{haccou}, \cite{pjstab}, as $K$ grows. Below, $q_z = q_z(K)$ denotes the extinction probability of the process with parameters frozen at population size $z$, started from one newly born ancestor, and $m_z = m_z(K)<1$ is the expected number of offspring per individual for $K$ given in the same population, conditioned to die out. 

Before formulating the extinction time theorem, we give a lemma, of some independent interest, about subcritical Galton-Watson processes, which are regular in the sense that there is an $r>1$ such the process reproduction generating function $f$ satisfies $f(r)=r$. The reader may note that this always is the case for a subcritical process which is at bottom a supercritical one, but conditioned to die out. Indeed, if $q<1$ denotes the extinction probability of the supercritical process, then the conditioned generating function $h$ will satisfy $h(s) = f(sq)/q$, so that $h(1/q) = f(1)/q = 1/q$. For a related result cf. also Lemma 3.2 of \cite{dembo} (telling that the moments of the total progeny of a subcritical Galton-Watson process in random environments are finite together with the moments of its reproduction distribution).

\begin{lemma} Consider a subcritical Galton-Watson process $\{\zeta_n, \zeta_0=1\}$ with the generating function $f$, $f(r) = r$ for an $r>1$. Then,  the probability generating function $g$ of the total progeny $\eta=\sum_{n=0}^\infty \zeta_n$, $g(s)={\Bbb E}[s^\eta]]$, will converge for some $s, 1<s<r$.
\end{lemma}

\begin{proof}
As well established, $g$ will satisfy $g(s)=sf\circ g(s), 0\le s\le 1$, cf. \cite{harris}, \cite{haccou} or any branching process monograph. It is bounded and strictly increasing on the unit interval, and so has an inverse $g^{-1}$ on $[g(0), g(1)], g(0)=0, g(1)=r$, since $g(1)>1$ and $f$ has precisely two fixpoints, 1 and $r$. Clearly,
$$g^{-1}(u) = u/f(u).$$
The right hand side is defined on $[0,r]$ and increases strictly from zero to a maximum at the point $v, 1<v<r$ where $f(v) = vf'(v)$. Hence, $g$ is well defined and bounded on the interval $[0,g^{-1}(v]]$, where  the right end point equals the asked for $s=v/f(v)>1$. 
\end{proof}

\begin{theorem}
Beyond earlier assumptions, in particular the monotonicity of frozen processes and the stabilisation of individual life and reproduction laws as $K\to\infty$, assume that for any $0<d<1$, $q_{dK}-q_1 = o(1/K)$ and $\lim_{K\to\infty}m_1(K) =m_1<1$. Let $0<a<b(1-m_1), b<1$.  Consider the process, started at $aK$ and write $\bar{Z}$ for its maximum. Then, as $K\to \infty$, $T|\bar{Z}\le bK = O(\log K)$.

\end{theorem}

\begin{proof}
On the same probability space we define processes $Z^{(k)}$, all sharing starting  size with $Z$, but parameters frozen at the population size $k$. This can be so done that, for all $t$, $Z^{(y)}_t \le Z_t \le Z^{(1)}_t$ on the set where $\bar{Z}\le y$, and also $Z^{(y)}_0 = Z_0 = Z^{(1)}_0 = x$. We write $Q^{(k)}= \{Z^{(k)}\to 0\}$ and $T^{(k)}$ for the corresponding (possibly infinite) extinction times. Bars indicate process maxima throughout,

\[
{\Bbb P}(T \le t, \bar{Z}\le y) = {\Bbb P}( Z_t=0, \bar{Z}\le y) \ge  {\Bbb P}( Z^{(1)}_t=0, \bar{Z}^{(1)}\le y) =\] \[{\Bbb P}( Z^{(1)}_t=0, \bar{Z}^{(1)}\le y; Q^{(1)}) =  {\Bbb P}( Z^{(1)}_t=0; Q^{(1)}) - {\Bbb P}( Z^{(1)}_t=0, \bar{Z}^{(1)}> y; Q^{(1)}) \]\[ \ge {\Bbb P}( Z^{(1)}_t=0; Q^{(1)}) - {\Bbb P}( \bar{Z}^{(1)}> y; Q^{(1)}) .
\]

But ${\Bbb P}( \bar{Z}^{(1)}> y; Q^{(1)}) =  {\Bbb P}(\bar{Z}^{(1)}>y|Q^{(1)}){\Bbb P}(Q^{(1)}).$
Now, if $\{Z^{(1)}_u\}$ were a Galton Watson process, we could conclude from \cite{Lindvall},  $q_1= {\Bbb P}_1(Q^{(1)})$, that
\begin{equation}
 {\Bbb P}(\bar{Z}^{(1)}>y|Q^{(1)}) \le \frac{q_1^{y-x}-q_1^y}{1-q_1^{y}} = O(q_1^{y-x}),
\end{equation} 
tending to zero for $K\to\infty$, if $x=aK, y=bK, 0<a<b<1$ without further ado.

In the general case, many generations can overlap and we can only assert that the maximum of a subcritical  process cannot exceed the total progeny $Y^{1}$ of the ancestors, which in its turn is the sum of the i.i.d. total progenies of each of the ancestors, to be denoted by $\eta_i$, and of course coinciding with the total progenies of the embedded Galton-Watson processes. However, by the lemma above, we know that these have all moments finite. In particular ${\Bbb E}[\eta_i] =1/(1-m_1(K))\to 1/(1-m_1)$,  by the stabilisation of processes, as $K\to\infty$. The law of large numbers (or central limit theorem) applies to show that
\begin{equation}
 {\Bbb P}_x(\bar{Z}^{(1)}>y|Q^{(1)}) \le  {\Bbb P}(\sum_{i=1}^x \eta_i >y) =\epsilon_K \to 0, 
\end{equation}
precisely under the stated conditions. 


Further,
\[{\Bbb P}(\bar{Z}\le y) \le {\Bbb P}(\bar{Z}^{(y)}\le y) = {\Bbb P}(Q^{(y)})(1- {\Bbb P}(\bar{Z}^{(y)}> y|Q^{(y)})) \le q_y^x.
\]
In other words,
\[ {\Bbb P}(T \le t |\bar{Z}\le y) \ge q_1^x( {\Bbb P}( T^{(1)} \le t| Q^{(1)}) -\epsilon_K)/q_y^x,\]
where we keep in mind that $x=aK, y=bK$, the starting point $x$ is subsumed, and extinction probabilities also depend upon $K$. Since
$$\Big(\frac{q_1}{q_{bK}}\Big)^{aK} \to 1,$$
as $K\to \infty$, we can conclude that  asymptotically $T = O(\log K)$. 

\end{proof}

\begin{acknowledgements}
This research has been supported by the Australian Research Council Grant DP120102728. 
\end{acknowledgements}


\begin{thebibliography}{}

\bibitem{asmus} Asmussen, S. and Hering, H., {\em Branching
    Processes}. Birkh\"auser. Boston (1983).
 \bibitem{bb}  Borde-Boussion, A.-M., Stochastic
demographic models: age of a population. {\it Stoch. Proc. Appl.} {\bf
  35}, 279--291 (1990).

\bibitem{meleard} Champagnat N., Ferriere R., and M\'el\'eard, S.,
From individual stochastic processes to macroscopic models in adaptive
evolution. {\it Stoch. Models} {\bf 24}, 2--44  (2008).

\bibitem{chapman} Chapman, S.~J. et al.,  A nonlinear model of age and
  size-structured populations with applications to cell cycles. {\it
    ANZIAM J.}  {\bf 49}, 151--169 (2007).

\bibitem{CL} Chigansky, P. and Liptser, R., Moderate
deviations for a diffusion type process in random environment.
{\em Th. Prob. Appl.} {\bf 54}, 29--50 (2010).

\bibitem{cushing}  Cushing, J.~M., Existence and stability of
  equilibria in age-structured population dynamics. {em
    J. Math. Biology} {\bf 20}, 259--276 (1984).

\bibitem{dawson} Dawson, D.~A., Measure-valued Markov processes. {\em
    \'Ecole d'Et\'e de Probabilit\'es de Saint-Flour XXI. Lecture
    Notes in Math.} {\bf 1541}, Springer, Berlin (1993).

\bibitem{dembo} Dembo A., Peres Y., Zeitouni O.,  Tail estimates for one-dimensional
random walk in random environment. {\em Comm. Math. Phys.} {\bf 181},
667–-683 (1996).

\bibitem{dl}  Dieckmann, U. and Law, R. The dynamical theory of
  coevolution: a derivation from stochastic ecological
  processes. {\it  J. Math. Biology} {\bf 34}, 579--612 (1996). 
\bibitem{odo1} Diekmann, O., Gyllenberg, M.,  Metz, J,
  On the formulation and analysis of general deterministic structured
  population models I. Linear Theory {\it J. Math. Biology} {\bf 36},
  349-388  (1998).

\bibitem{odo2} Diekmann, O., Gyllenberg, M.,  Metz, J. , et al.,
  On the formulation and analysis of general deterministic structured
  population models II. Nonlinear Theory {\it J. Math. Biology} {\bf
    43}, 157--189  (2001). 

\bibitem{odwyer}  O'Dwyer, J. P. et al., An integrative framework for
  stochastic, size-structured community assembly. {\it
    Proc. Nat. Acad. Sci.}  {\bf 106},  6170-6175 (2009).  

\bibitem{eandk} Ethier, S.~N. and Kurtz, T.~G.,  {\it Markov
    Processes.} Wiley, New York (1986).
\bibitem{FW} Freidlin, M.~I. and Wentzell, A.~D.,. {\em Random
    Perturbations of Dynamical Systems.}  Springer-Verlag, New York (1998).
\bibitem{metz} Geritz S. A. H., Kisdi, \'E.,  Mesz\'ena, G., and  Metz
  J. A. J., Evolutionarily singular strategies and the adaptive growth
  and branching of the evolutionary tree. {\it Evol. Ecol.} {\bf 12},
  35--57 (1998).
\bibitem{haccou}  Haccou, P., Jagers, P., and Vatutin, V.~A., {\em
    Branching Processes: Variation, Growth, and Extinction of
    Populations}. Cambridge Univ. Press, Cambridge  (2005).
\bibitem{trudy} Hamza, K., Jagers, P., and Klebaner, F.~K.,  The age
  structure of population-dependent general branching processes in
  environments with a high carrying capacity. {\it Proc. Steklov
    Inst. Math.} {\bf 282}, 90--105 (2013).
\bibitem{harris} Harris, T.~E., {\it The Theory of Branching Processes.} Springer (1963), Dover (1989).
\bibitem{pjbook} Jagers, P., {\it Branching Processes
with Biological Applications.} Wiley. Chichester (1975).
\bibitem{pjstab} Jagers, P., Stabilities and instabilities in
  population dynamics. {\em J. Appl. Prob.} {\bf 29}, 770--780 (1992).

\bibitem{JK} Jagers, P. and Klebaner, F.~C., 
 Population-Size-Dependent and Age-Dependent Branching Processes,
  {\it Stoch. Proc. Appl.} {\bf 87}, 235--254 (2000).

\bibitem{pnas} Jagers, P., Klebaner, F. C., and Sagitov, S., On the
  path to existence. {\em Proc. Nat. Acad. Sci.} 104, 6107-6111 (2007).

\bibitem{lingering} Jagers, P., and Klebaner F.~C.,
  Population-size-dependent, age-structured branching processes linger
  around their carrying capacity.  {\it J. Appl. Prob.}  {48A},
  249-260 (2011).

\bibitem{lageras}  Jagers, P. and Nordvall Lagerås, A., General
  branching processes conditioned on extinction are still branching
  processes {\it Elect. Comm. Probab.} 13:51 (2008).

\bibitem{jn}  Jagers, P. and Nerman, O., The asymptotic composition of
supercritical multi-type branching populations. {\it Springer Lecture
  Notes in Mathematics} {\bf 1626}, 40 - 54 (1996).

\bibitem{Jakub}  Jakubowski, A. On the Skorokhod topology. {\it
    Ann. Inst. H. Poincar\'e} {\bf B22}, 263--285 (1986).

\bibitem{Janson} Janson, S., Large deviation inequalities for sums of
  indicatior variables. Uppsala U. Dep. Mathematics, Tech. Report 34 (1994).

\bibitem{Kallenberg} Kallenberg O., {\it Foundations of Modern
    Probability.}  2nd ed., Springer, Berlin etc. (2002).

\bibitem{kleb84} Klebaner  F.C., Geometric rate of growth in
  population size dependent branching processes. {\it J. Appl. Prob.}
  21, 40--49 (1984).
 \bibitem{fkbook}  Klebaner F.~C., {\it Introduction to Stochastic
     Calculus with Applications}, 2nd.~ed. Imperial College Press,
   London (2005).
\bibitem{ksvhj}  Klebaner, F.~C., Sagitov, S., Vatutin, V.~A., Haccou,
  P., and Jagers, P., Stochasticity in the adaptive dynamics of
  evolution: the bare bones. {\em J. Biol. Dyn.} 5,147--162 (2011).
\bibitem{Lindvall} Lindvall T., On the maximum of a branching process.
  {\it Scand. J. Statist.}  3, 209--214 (1976).
\bibitem{meltran}  M\'el\'eard, S. and Tran, V.~C., Trait substitution
  sequence process and canonical equation for age-structured
  populations. {\em J. Math. Biol.} {\bf 58}, 881--921  (2009).
\bibitem{metivier} M\'etivier, M., Weak convergence of measure valued
  processes using Sobolev imbedding techniques. {\em Lect.
Notes in Maths.} {\bf 1236},  172--183. Springer, Berlin (1987).
\bibitem{metz96} Metz, J.A.J. Geritz, S.A.H., Mesz\'ena, G., et al.,
  Adaptive Dynamics, a geometrical study of nearly faithful
  reproduction. In: S.J. van Strien and S.M: Verdyan Lunel (eds.) {\it
    Stochastic and Spatial Structures of Dynamical Systems},
  183--231. North Holland, Amsterdam (1996).


\bibitem{Oe} Oelschl\"ager, K., Limit theorems for age-structured
populations. {\it Ann. Prob.} {\bf 18}, 290--318 (1990).
\bibitem{tran} Tran, V. C., Large population limit and time behaviour
  of a stochastic particle model describing an age-structured
  population. {\it ESAIM: Probability and Statistics} 12, 345--386 (2008).
\bibitem{waxman}  Waxman, D. and Gavrilets, S., 20 Questions on
  adaptive dynamics. {\it J. Evol. Biol.} {\bf 18}, 1139--1154 (2005).
\bibitem{webb} Webb, G.F.,  {\it Theory of Nonlinear Age-Dependent
    Population Dynamics.} Dekker, New York (1985).

\end{thebibliography}
\end{document}